\documentclass[a4paper,11pt]{article}
\usepackage{amsfonts}
\usepackage{pstricks}
\usepackage{pst-node,pst-tree,pst-eps,xspace}
\IfFileExists{pst-beta.tex}
             {\input{pst-beta.tex}}
             {\RequirePackage{pst-tree}}

\usepackage{latexsym}
\usepackage{amsthm}
\usepackage{fancyheadings}
\usepackage{amsmath}
\usepackage{graphicx}
\usepackage{float}

\usepackage{tikz}
\usetikzlibrary{matrix,arrows}

\restylefloat{table}
\newcommand{\s}{\mathfrak{S}}
\def\maj{{\scriptstyle \mathsf{MAJ}}}
\def\des{{\small \mathsf{des}}}
\def\ides{{\small \mathsf{ides}}}
\def\red{{\small \mathsf{red}}}
\def\mad{{\scriptstyle \mathsf{MAD}}}
\def\mak{{\scriptstyle \mathsf{MAK}}}
\def\den{{\scriptstyle \mathsf{DEN}}}
\def\inv{{\scriptstyle \mathsf{INV}}}
\def\stat{{\scriptstyle \mathsf{STAT}}}
\def\st{{\scriptstyle \mathsf{ST}}}
\def\stp{{\scriptstyle \mathsf{ST'}}}
\def\adj{{\small \mathsf{adj}}}
\def\Fi{{\scriptstyle \mathsf{F}}}

\newcommand{\vinc}[3]{
\begin{tikzpicture}[baseline = (X.base)]
	\useasboundingbox (0.1,0) rectangle (#1*0.23,0.1);
	\foreach \x/\y in {#2}
	{
		\draw (\x*0.2,0) node (X) {$\y$};
	}
	
	\foreach \z in {#3}
	{
		\ifnum 0<\z
			\ifnum \z<#1
				\draw[thick] (\z*0.2-0.07,-0.19) -- (\z*0.2+0.27,-0.19);
			\fi
		\fi
		
		\ifnum 0=\z
			\draw[thick] (0.07,0.1) -- (0.07,-0.19) -- (0.21,-0.19);
		\fi
		
		\ifnum \z=#1
			\draw[thick] (\z*0.2+0.14,0.1) -- (\z*0.2+0.14,-0.19) -- (\z*0.2,-0.19);
		\fi
	}
\end{tikzpicture}
}

\def\ppart{{ \mathsf{ppart}}}
\def\wpart{{ \mathsf{wpart}}}

\def\Exp{{ \mathsf{exp}}}
\def\flat{{ \mathsf{flat}}}

\newtheorem{Le}{Lemma}

\newtheorem{Co}{Corollary}

\newtheorem{The}{Theorem}

\theoremstyle{definition}

\newtheorem{Fact}{Fact}
\newtheorem{Exam}{Example}

\author{Sergey {\sc Kitaev}\\
{\small Department of Computer and Information Sciences, University of Strathclyde}\\
{\small Livingstone Tower, 26 Richmond Street, Glasgow G1 1XH, UK}\\
{\small \tt sergey.kitaev@cis.strath.ac.uk}\\
\\
Vincent {\sc Vajnovszki}\\ 
{\small LE2I, Universit\'e de Bourgogne}\\
{\small BP 47870, 21078 Dijon Cedex, France}\\
{\small \tt vvajnov@u-bourgogne.fr}
}

\title{Mahonian $\stat$ on words}
\begin{document}
\maketitle

\begin{abstract}

In 2000, Babson and Steingr{\'\i}msson introduced the notion of what is now known as a
permutation vincular pattern, and based on it they re-defined known Mahonian statistics 
and introduced new ones,  proving or conjecturing their Mahonity.  These conjectures were 
proved by Foata and Zeilberger in 2001, and by
Foata and Randrianarivony in 2006.

In 2010, Burstein refined some of these results by 
giving a bijection between 
permutations with a fixed value for the major index and those with the same value for 
$\stat$, where $\stat$ is 
one of the statistics defined and proved to be Mahonian in the 2000 
Babson and Steingr{\'\i}msson's paper. Several other statistics are preserved as well 
by Burstein's bijection.

At the Formal Power Series and Algebraic Combinatorics Conference 
(FPSAC) in 2010, Burstein asked whether his bijection has other interesting properties. In this paper, we not only show that 
Burstein's bijection preserves the Eulerian statistic $\ides$, but also use this fact, along with the bijection itself, to prove Mahonity of the statistic $\stat$ on words we introduce in this paper.  The words statistic $\stat$ introduced by us here addresses a natural question on existence of a Mahonian words analogue of $\stat$ on permutations. 
While proving Mahonity of our $\stat$ on words, we prove a more general 
joint equidistribution result involving  two six-tuples of statistics on (dense) words, where Burstein's bijection plays an important role.  
\end{abstract}

\section{Introduction}
 
In \cite{BabSteim}, the notion of what is now known as a {\em vincular 
pattern}\footnote{Such patterns are called {\em generalized patterns} in \cite{BabSteim}.} on
permutations was introduced, and it was shown that almost all known {\em Mahonian} permutation statistics 
(that is, those statistics that are distributed as $\inv$ or as $\maj$ to be defined in Section~\ref{sec2})
can be expressed as combinations of vincular patterns. 
The authors of \cite{BabSteim} also introduced some new vincular pattern-based 
permutation statistics, 
showing that some of them are Mahonian and conjecturing that others  
are Mahonian as well.
These conjectures were proved later in 
\cite{Foata_Zeilberger,Foata_Randrianarivony},
and recently, alternative proofs based on {\em Lehmer code transforms} were
given in \cite{Vaj_13}.

Three statistics expressed in terms of vincular pattern
combinations in \cite{BabSteim} 
(namely, $\mak$, $\mad$ and $\den$) are known to be Mahonian not only on permutations, but 
also on words (see \cite[Theorem 5]{ClaSteimZeng}); 
more precisely, for any word $v$, the three statistics are 
distributed as $\inv$ on the set of rearrangements of the letters of $v$.

One of the statistics defined and shown to be Mahonian in \cite{BabSteim}
is $\stat$. Generalizing a result in \cite{Foata_Zeilberger},
Burstein \cite{Bur}  shown the equidistribution of $\stat$ and $\maj$
together with other statistics by means of an involution $p$ on the set of permutations. At the Formal Power Series and Algebraic Combinatorics Conference 
(FPSAC) in 2010, Burstein asked whether $p$ has other interesting properties. 

In this paper, we not only show that 
$p$ preserves the Eulerian statistic $\ides$ (which is not preserved, e.g. by the bijection $\Phi$ on words \cite{ClaSteimZeng} mapping $\mad$ to $\inv$), but also use this fact, along with $p$ itself, to prove Mahonity of the statistic $\stat$ on words introduced in Subsection~\ref{vinc-pat} (see relation (\ref{def_stat})).  The words statistic $\stat$ introduced by us in this paper addresses a natural question on existence of a Mahonian words analogue of $\stat$ on permutations. 
While proving Mahonity of our $\stat$ on words, we prove a more general 
joint equidistribution result involving  two six-tuples of statistics on (dense) words, where the bijection $p$ plays an important role (see Theorems~\ref{thm1} and~\ref{thm2} in Section~\ref{ext-of-p-to-words}).


\section{Preliminaries}\label{sec2}

We denote by $[n]$ the set $\{1,2,\ldots,n\}$, by 
$\s_n$ the set of permutations of $[n]$, and by
$[q]^n$ the set of length $n$ words 
over the alphabet $[q]$. Clearly $\s_n\subset [q]^n$ for $q\geq n>1$.
A word $v$ in $[q]^n$ is said to be 
{\it dense}
if each letter in $[q]$ occurs at least once in $v$. Dense words are also called {\em multi-permutations}.

\subsection{Statistics}

A {\em statistic} on $[q]^n$ (and thus on $\s_n$) is an association of an integer
to each word in $[q]^n$. Classical examples of statistics are:
\begin{itemize}
\item[]
$
\inv\, v =\text{card}\, \{(i,j)\ :\ 1\leq i<j\leq n, v_i>v_j\},
$
\item[]
$
\displaystyle 
\maj\, v = \mathop{\sum_{1\leq i <n}}_{v_i>v_{i+1}} i,
$
\item[]
$
\displaystyle 
\des\, v = \text{card}\, \{i\ :\ 1\leq i<n, v_i>v_{i+1}\},
$
\end{itemize}
where $v=v_1v_2\ldots v_n$ is a length $n$ word. 
For example, $\inv(31425)=3$, $\maj(3314452)=8$, and $\des(8416422)=4$.

For a word $v$ and a letter $a$ in $v$, 
other than the largest one in $v$,
let us denote by $next_v(a)$ the smallest letter in $v$ larger
than $a$. With this notation, we define

\begin{itemize}
\item[]
$
\displaystyle 
\ides\, v = \text{card}\, \{a :\ {\rm there\ are}\ i\ {\rm and}\ j,1\leq i<j\leq n,\ {\rm with}\ 
v_i=next_v(a)\ {\rm and}\ v_j=a \}.
$
\end{itemize}

\noindent
Clearly, when $v$ is a permutation, $\ides\, v$ is simply $\des\, v^{-1}$, where $v^{-1}$ is 
the inverse of $v$. For example, $\ides(144625)=2$, 
and the corresponding values for $a$ are 2 and 5.

For a set of words $S$, two statistics $\st$ and $\st'$ have the same distribution 
(or are equidistributed) on $S$
if, for any $k$, 
$$
\mathrm{card}\{v\in S: \st\,v=k\}=
\mathrm{card}\{v\in S: \stp\,v=k\},
$$
and it is well-known that $\inv$ and $\maj$ have the same distribution
on both, the set of permutations and that of words.

A {\it multi-statistic} is simply a tuple of statistics.


\subsection{Vincular patterns}\label{vinc-pat}

Let $1\leq r\leq q$ and $1\leq m\leq n$, and let $v\in [r]^m$ be a dense word.
One says that  $v$ occurs as a (classical) pattern in $w=w_1w_2\cdots w_n\in[q]^n$
if there is a sequence $1\leq i_1<i_2<\cdots<i_m\leq n$ such that
$w_{i_1}w_{i_2}\cdots w_{i_m}$ is order-isomorphic to
$v$. For example, $1231$ occurs as a pattern in $6214562$,
and the three occurrences of it are $2452$, $2462$ and $2562$.

Vincular patterns were introduced in the context of permutations in \cite{BabSteim}
and they were extensively studied since then (see 
Chapter 7 in \cite{Kit} for a comprehensive description of results on these patterns). 
Vincular patterns 
generalize classical patterns and they are defined as follows:
\begin{itemize}
\item Any pair of two adjacent letters may now be underlined, which means that the 
corresponding letters in the permutation must be 
adjacent\footnote{The original notation for vincular patterns uses dashes: 
the absence of a dash between two letters of a pattern means that these letters 
are adjacent in the permutation.}.  For example, the pattern 
$\vinc{3}{1/2,2/1,3/3}{2}$ occurs in the permutation 425163 four times, namely, as 
the subsequences $425$, $416$, $216$ and $516$. Note that, the subsequences 426 and 213 are {\em not} 
occurrences of the pattern because their last two letters are not adjacent in the permutation. 
\item If a pattern begins (resp., ends) with a 
hook\footnote{In the original notation the role of hooks was played by square brackets.} 
then its occurrence is required to begin (resp., end) with the leftmost (resp., rightmost) 
letter in the permutation. For example, there are two occurrences of the pattern 
$\vinc{3}{1/2,2/1,3/3}{0,2}$ in the permutation $425163$, which are the subsequences 
$425$ and $416$.
\end{itemize}

The notion of a vincular pattern is naturally extended to words. 
For example, in the word $6214562$, $645$ is an occurrence of the pattern $\vinc{3}{1/3,2/1,3/2}{2}$, and
$262$ is that of $\vinc{3}{1/1,2/2,3/1}{2,3}$ .

For a set of patterns $\{p_1,p_2,\ldots\}$ we denote by 
$(p_1+p_2+\ldots)$ the statistic giving the total number of occurrences of the patterns
in a permutation. It follows from definitions that

\begin{equation}
\maj\, v=(\vinc{3}{1/1,2/3,3/2}{2} +\vinc{3}{1/1,2/2,3/1}{2} +\vinc{3}{1/2,2/3,3/1}{2} +\vinc{3}{1/2,2/2,3/1}{2} 
         +\vinc{3}{1/3,2/2,3/1}{2} +\vinc{2}{1/2,2/1}{1})\,v.
\label{gen_maj}
\end{equation}

A vincular pattern of the form $\vinc{3}{1/u,2/v,3/x}{2}$,
with $\{u,v,x\}=\{1,2,3\}$, is determined by the relative order
of $u$, $v$ and $x$. For example, $\vinc{3}{1/2,2/1,3/3}{2}$ is determined by $v<u<x$, and 
$\vinc{3}{1/3,2/2,3/1}{2}$ by $x<v<u$.

An {\it extension} of a vincular pattern $\vinc{3}{1/u,2/v,3/x}{2}$, $\{u,v,x\}=\{1,2,3\}$,
is the combination of the vincular patterns obtained by replacing an order relation involving 
$u$ (possibly both of them if there are two) by
its (their) weak counterpart. For example, 

\begin{itemize}
\item the unique extension of $\vinc{3}{1/1,2/3,3/2}{2}$ is 
      $(\vinc{3}{1/1,2/3,3/2}{2}+\vinc{3}{1/1,2/2,3/1}{2})$; and 
\item the three extensions of $\vinc{3}{1/2,2/3,3/1}{2}$ are:
    \begin{itemize}
    \item[] $(\vinc{3}{1/2,2/3,3/1}{2}+\vinc{3}{1/1,2/2,3/1}{2})$,
    \item[] $(\vinc{3}{1/2,2/3,3/1}{2}+\vinc{3}{1/2,2/2,3/1}{2})$, and
    \item[] $(\vinc{3}{1/2,2/3,3/1}{2}+\vinc{3}{1/1,2/2,3/1}{2}+\vinc{3}{1/2,2/2,3/1}{2})$.
    \end{itemize}
\end{itemize}

An extension of a vincular pattern $\vinc{3}{1/u,2/v,3/x}{1}$ is defined similarly, and 
an extension of $(p_1+p_2+\ldots)$ is the statistic obtained by extending some 
of $p_i$'s.

With these notations, the definition of $\maj$ in (\ref{gen_maj}) is an extension of 
$\maj$ defined on $\s_n$:

\begin{equation}
\maj\, v=(\vinc{3}{1/1,2/3,3/2}{2} +\vinc{3}{1/2,2/3,3/1}{2} +\vinc{3}{1/3,2/2,3/1}{2} +\vinc{2}{1/2,2/1}{1})\,v.
\label{perm_maj}
\end{equation}

The statistic $\stat$ on permutations was introduced and shown to be Mahonian
in \cite{BabSteim}, i.e. distributed as $\maj$; $\stat$ is 
defined as:
$$
\stat\,\pi=(\vinc{3}{1/2,2/1,3/3}{1} +\vinc{3}{1/1,2/3,3/2}{1} +\vinc{3}{1/3,2/2,3/1}{1}  
+\vinc{2}{1/2,2/1}{1})\,\pi.
$$
An extension of $\stat$ to words, where repeated letters are allowed, is to extend:

\begin{itemize}
\item $\vinc{3}{1/2,2/1,3/3}{1}$ as  $(\vinc{3}{1/2,2/1,3/3}{1} + \vinc{3}{1/2,2/1,3/2}{1})$, and 
\item $\vinc{3}{1/1,2/3,3/2}{1}$ as  $(\vinc{3}{1/1,2/3,3/2}{1} + \vinc{3}{1/1,2/2,3/1}{1})$,
\end{itemize}
and thus to define the statistic $\stat$ on words as:
\begin{equation}
\stat\,v=( \vinc{3}{1/2,2/1,3/3}{1} + \vinc{3}{1/2,2/1,3/2}{1} +
           \vinc{3}{1/1,2/3,3/2}{1} + \vinc{3}{1/1,2/2,3/1}{1} +
	   \vinc{3}{1/3,2/2,3/1}{1} +
	   \vinc{2}{1/2,2/1}{1})\,v.
\label{def_stat}
\end{equation}
In what follows, we will use 
this definition which seems to be sporadic and not any better than any other
possible extension of $\stat$ from permutations to words.
However, a consequence of Theorem \ref{the_eq_w} is that this extension has the same distribution as 
$\maj$ on words, and experimental tests show that no other extension
(in the sense specified above) does so.

\section{The bijection $p$ on $\s_n$}\label{sec3}

Now we present the involution on $\s_n$ introduced in \cite{Bur} 
which maps a permutation with a given value for $\maj$ to one with the same 
value for $\stat$, and show that among other statistics, it preserves $\ides$.

For three integers $a\leq x\leq b$, the {\it complement} of $x$ with respect to 
the interval $\{a,a+1,\ldots,b\}$ is simply the integer $b-(x-a)$.

For a $\pi=\pi_1\pi_2\ldots\pi_n\in\s_n$, let us define
\begin{itemize}
\item $\pi'\in\s_n$ by $\pi'_1=\pi_1$, and for $i\geq 2$,
      $\pi'_i$ is the complement of $\pi_i$ with respect to 
      \begin{itemize}
      \item $\{\pi_1+1,\pi_1+2,\ldots,n\}$, if $\pi_i>\pi_1$, and 
      \item $\{1,2,\ldots,\pi_1-1\}$, if 
            $\pi_i<\pi_1$;
      \end{itemize}
\item  $\pi''\in\s_n$ by $\pi''_1=\pi'_1=\pi_1$ and $\pi''_i=\pi'_{n-i+2}$.
\end{itemize}

Clearly, the map $\pi\mapsto \pi''$ is a bijection on $\s_n$. In fact,  $p$ is an involution, that is 
$p(p(\pi))=\pi$. See Figure \ref{pis} for an example.

\begin{figure}
\begin{center}
\begin{tabular}{ccc}
\unitlength=3mm\begin{picture}(6,6)
\put(0.,0.){\line(1,0){6}}
\put(0.,1.){\line(1,0){6}}
\put(0.,2.){\line(1,0){6}}
\put(0.,3.){\line(1,0){6}}
\put(0.,4.){\line(1,0){6}}
\put(0.,5.){\line(1,0){6}}
\put(0.,6.){\line(1,0){6}}
\put(0.,0.){\line(0,1){6}}
\put(1.,0.){\line(0,1){6}}
\put(2.,0.){\line(0,1){6}}
\put(3.,0.){\line(0,1){6}}
\put(4.,0.){\line(0,1){6}}
\put(5.,0.){\line(0,1){6}}
\put(6.,0.){\line(0,1){6}}
\put(0.5,3.5){\circle*{0.3}} 
\put(1.5,4.5){\circle*{0.3}} 
\put(2.5,1.5){\circle*{0.3}} 
\put(3.5,5.5){\circle*{0.3}} 
\put(4.5,2.5){\circle*{0.3}} 
\put(5.5,0.5){\circle*{0.3}} 
\end{picture}
&
\unitlength=3mm\begin{picture}(6,6)
\put(0.,0.){\line(1,0){6}}
\put(0.,1.){\line(1,0){6}}
\put(0.,2.){\line(1,0){6}}
\put(0.,3.){\line(1,0){6}}
\put(0.,4.){\line(1,0){6}}
\put(0.,5.){\line(1,0){6}}
\put(0.,6.){\line(1,0){6}}
\put(0.,0.){\line(0,1){6}}
\put(1.,0.){\line(0,1){6}}
\put(2.,0.){\line(0,1){6}}
\put(3.,0.){\line(0,1){6}}
\put(4.,0.){\line(0,1){6}}
\put(5.,0.){\line(0,1){6}}
\put(6.,0.){\line(0,1){6}}
\put(0.5,3.5){\circle*{0.3}} 
\put(1.5,5.5){\circle*{0.3}} 
\put(2.5,1.5){\circle*{0.3}} 
\put(3.5,4.5){\circle*{0.3}} 
\put(4.5,0.5){\circle*{0.3}} 
\put(5.5,2.5){\circle*{0.3}} 
\end{picture}
&
\unitlength=3mm\begin{picture}(6,6)
\put(0.,0.){\line(1,0){6}}
\put(0.,1.){\line(1,0){6}}
\put(0.,2.){\line(1,0){6}}
\put(0.,3.){\line(1,0){6}}
\put(0.,4.){\line(1,0){6}}
\put(0.,5.){\line(1,0){6}}
\put(0.,6.){\line(1,0){6}}
\put(0.,0.){\line(0,1){6}}
\put(1.,0.){\line(0,1){6}}
\put(2.,0.){\line(0,1){6}}
\put(3.,0.){\line(0,1){6}}
\put(4.,0.){\line(0,1){6}}
\put(5.,0.){\line(0,1){6}}
\put(6.,0.){\line(0,1){6}}
\put(0.5,3.5){\circle*{0.3}} 
\put(1.5,2.5){\circle*{0.3}} 
\put(2.5,0.5){\circle*{0.3}} 
\put(3.5,4.5){\circle*{0.3}} 
\put(4.5,1.5){\circle*{0.3}} 
\put(5.5,5.5){\circle*{0.3}} 
\end{picture}
\\
(a) & (b) & (c)
\end{tabular}
\end{center}
\caption{\label{pis}
The permutations: (a) $\pi=452631$, (b) $\pi'=462513$, and (c) $\pi''=p(\pi)=431526$.}
\end{figure}
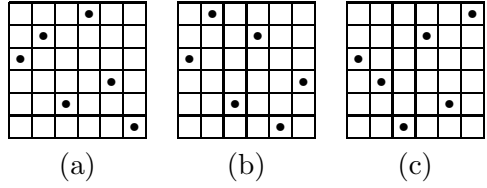

Also, in \cite{Bur}  is proved that, for any $\pi=\pi_1\pi_2\ldots\pi_n\in\s_n$, 
the $5$-tuple $(\adj,\des,\Fi,\maj,\stat)\,\pi$ is equal to 
$(\adj,\des,\Fi,\stat,\maj)\,p(\pi)$, where  

\begin{itemize}
\item $\adj\, \pi = \text{card}\, \{i\ :\ 1\leq i\leq n, \pi'_i=\pi'_{i+1}+1\}$, 
      where $\pi'=\pi0$, and
\item $\Fi\,\pi=\pi_1$.
\end{itemize}

Below, we will use the following result.

\begin{Le}\label{lemma1}
For any $\pi=\pi_1\pi_2\ldots\pi_n\in\s_n$, $\ides\, \pi=\ides\, p(\pi)$.
\label{le_th}
\end{Le}

\proof
An integer $a$ is an occurrence of an $\ides$ in $\pi$ if there are $i<j$ such that 
$\pi_i=a+1$ and $\pi_j=a$.
Clearly, if $\pi_1>1$, then $\pi_1-1$ is an occurrence of an $\ides$ in both $\pi$ and 
$\sigma=p(\pi)$. And $a\neq\pi_1-1$ is an occurrence of an $\ides$ in $\pi$ 
if and only if so is the element in position $n-\pi^{-1}(a+1)+2$ in $\sigma$, where $\pi^{-1}(a+1)$ is the position of the element $a+1$ in $\pi$.
\endproof

The following lemma, to be used later, follows directly from the proof of Lemma~\ref{lemma1}.

\begin{Le}
The number of $\ides$ in the interval
$\{1,2,\ldots,\pi_1-1\}$ is the same for $\pi$ and $p(\pi)$.
\label{le_by}
\end{Le}

\section{Interval partitions}

In this section, we define the notions of interval partitions of sets, permutations and words. We also define the notion of a word expansion.

\subsection{Interval partition of a set}
\label{def_psi}

An {\it interval partition of a set} $\{1,2,\ldots,n\}=[n]$
is a partition of this set, where each part is an interval
(i.e., a set consisting of successive integers), and the {\it size} of an interval partition 
is the number of its parts. For example, $\{\{1,2,3\},\{4,5\},\{6,7\}\}$ is an interval 
partition of size 3 of the set $[7]$.

For two interval partitions $R$ and $P$ of $[n]$, we say that
$R$ is a {\it refinement} of $P$, denoted by $R\subseteq P$, if each part of $R$ is a 
weak subset of a part of $P$. In particular, $P$ is a refinement of itself.
For example,  $R=\{\{1,2\},\{3\},\{4,5\},\{6\},\{7\}\}$  is a refinement of  $P=\{\{1,2,3\},\{4,5\},\{6,7\}\}$.

Note that any refinement $R$ of size $k+j$ of an interval partition $P$ of size $k$ can be encoded 
by an increasing sequence of $j$ numbers. For the last example, $R$ can be encoded by (2,4) because 
when creating the refinement, we scanned $P$ from left to right and have broken parts in the second 
and forth possible places. For the same interval partition  $P=\{\{1,2,3\},\{4,5\},\{6,7\}\}$, the 
encoding (1,3,4) would give the refinement $R=\{\{1\},\{2,3\},\{4\},\{5\},\{6\},\{7\}\}$. In 
general, for a partition of size $k$ of $\{1,2,\ldots,n\}$, we have $n-k$ possibilities to break a part 
and one possibility not to break anything. Thus, breaking parts, which  gives refinements, can be encoded uniquely by a 
possibly empty subsequence of increasing integers in $\{1,2,\ldots,n-k\}$.

Let $I_n$ denote the set of all interval partitions of $[n]$, and 
for $P\in I_n$,  we let
$$
I_n|P=\{R\in I_n : R\subseteq P\}.
$$
Now, for two same size interval partitions
$P,S\in I_n$, we define a map
$$
\psi_{P,S}:I_n|P\rightarrow I_n|S,
$$
which sends a refinement $R$ in $I_n|P$ to the refinement $T$ in $I_n|S$ such that $R$ and $T$ have the same 
encodings. It is straightforward to see that $\psi_{P,S}$ is a bijection, 
and its inverse is $\psi_{S,P}$.

\begin{Exam}
\label{P_S}
If 
$R=\{\{1\},\{2,3\},\{4\},\{5,6\}\}$, 
$P=\{\{1\},\{2,3\},\{4,5,6\}\}$ and 
$S=\{\{1,2\},\{3\},\{4,5,6\}\}$, then 
$T=\psi_{P,S}(R)=\{\{1,2\},\{3\},\{4\},\{5,6\}\}$.
\end{Exam}

\subsection{Interval partition of permutations}

The {\it interval partition of a permutation} $\pi\in\s_n$, denoted $\ppart(\pi)$,
is the interval partition of $[n]$ defined by:
$a$ and $a+1$ belong to the same part of $\ppart(\pi)$ if and only if
$a$ occurs to the left of $a+1$ in $\pi$. Thus, the partition of a permutation is given by its 
maximal increasing subpermutations of consecutive elements.
For example, if $\sigma=14235$ and $\pi =45123$, then
$\ppart(\pi)=\ppart(\sigma)=\{\{1,2,3\},\{4,5\}\}$.

Since an $\ides$ in $\pi$ is a value $a$ such that $a+1$ occurs to the left of $a$ in 
$\pi$, it follows that the size of $\ppart(\pi)$ is equal to $\ides\,\pi+1$, and
the next corollary is a consequence of Lemma \ref{le_th}.
\begin{Co}
For any $\pi\in\s_n$, the interval partitions of $\pi$ and that of $p(\pi)$
have the same size.
\label{cor}
\end{Co}

\subsection{Interval partition of words}

The {\it interval partition of a word} $v\in [q]^n$, denoted by $\wpart(v)$, 
is the interval partition 
$$\{p_1,p_2,\ldots, p_q\}
$$
of $[n]$
where the cardinality of each part $p_i$ is equal to the number of occurrences of the symbol
$i\in [q]$ in $v$, and empty parts, if any, are omitted.
Formally,  $p_i$ is given by

$$
p_i=\{a+1, a+2, \ldots, b\},
$$
with
$$a=|v|_1+|v|_2+\cdots +|v|_{i-1}, \ {\rm and}\ 
b=a+|v|_i,
$$ 
and the number of occurrences of each letter
in $v$ determines $\wpart(v)$.
For example, if $v=12112$ and $w=33111$, then 
$\wpart(v)=\wpart(w)=\{\{1,2,3\},\{4,5\}\}$.
In particular, when $v$ is a permutation in $\s_n$, $\wpart(v)=\{\{1\},\{2\},\ldots,\{n\}\}$.
See also Figure \ref{fig_Exp} for other examples.

\subsection{Words expansion} 
For $v\in [q]^n$, the {\it expansion} of $v$, denoted $\Exp(v)$,
is the unique permutation $\pi\in\s_n$ with $\pi_i<\pi_j$
if and only if either $v_i<v_j$, or  $v_i=v_j$ and $i<j$.
In particular, if $v$ is a permutation, then $\Exp(v)=v$.
For example,  
$\Exp(12112)=14235$ and $\Exp(22111)=45123$. 
We refer to Figure \ref{fig_Exp} for some other examples.
The following fact is easy to check.
\begin{Fact}
If $v$ is a dense word in $[q]^n$ and 
$\pi=\Exp(v)$, then $\wpart(v)$ is a refinement of  $\ppart(\pi)$.
\end{Fact}
%
Actually, $\Exp$ is a function from $[q]^n$ to $\s_n$, which is 
surjective if $q\geq n$, but not injective (again, see Figure~\ref{fig_Exp}). However, one can see that the following fact holds.

\begin{Fact}
The dense word $v$ is uniquely determined from $\Exp(v)$ and $\wpart(v)$. 
\label{fact_2}
\end{Fact}
\noindent
If $P$ is a refinement of $\ppart(\pi)$, 
we denote by $\flat_P(\pi)$ the unique word $v$ with  $\Exp(v)=\pi$ and $\wpart(v)=P$,
and so $\Exp(\flat_P(\pi))=\pi$.
Also, we will use the following fact 
which follows from the definitions of $\maj$ and $\stat$
given in relations (\ref{gen_maj}) and (\ref{def_stat}).

\begin{Fact} For any word $v$, we have $\maj\, v=\maj\,(\Exp(v))$ and 
$\stat\, v=\stat\,(\Exp(v))$.
\label{fact_3}
\end{Fact}

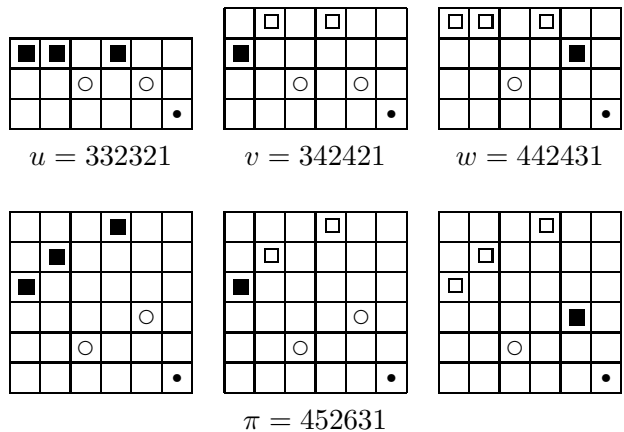
\begin{figure}
\begin{center}  
\begin{tabular}{ccc}
   \unitlength=4mm
\begin{picture}(6,3)
\put(0.,0.){\line(1,0){6}}
\put(0.,1.){\line(1,0){6}}
\put(0.,2.){\line(1,0){6}}
\put(0.,3.){\line(1,0){6}}
\put(0.,0.){\line(0,1){3}}
\put(1.,0.){\line(0,1){3}}
\put(2.,0.){\line(0,1){3}}
\put(3.,0.){\line(0,1){3}}
\put(4.,0.){\line(0,1){3}}
\put(5.,0.){\line(0,1){3}}
\put(6.,0.){\line(0,1){3}}
\put(0.3,2.25){\rule{0.2cm}{0.2cm}}
\put(1.3,2.25){\rule{0.2cm}{0.2cm}}
\put(3.3,2.25){\rule{0.2cm}{0.2cm}}
\put(2.5,1.5){\circle{0.4}}
\put(4.5,1.5){\circle{0.4}}
\put(5.5,0.5){\circle*{0.3}} 
\end{picture} & 
   \unitlength=4mm
\begin{picture}(6,3)
\put(0.,0.){\line(1,0){6}}
\put(0.,1.){\line(1,0){6}}
\put(0.,2.){\line(1,0){6}}
\put(0.,3.){\line(1,0){6}}
\put(0.,4.){\line(1,0){6}}
\put(0.,0.){\line(0,1){4}}
\put(1.,0.){\line(0,1){4}}
\put(2.,0.){\line(0,1){4}}
\put(3.,0.){\line(0,1){4}}
\put(4.,0.){\line(0,1){4}}
\put(5.,0.){\line(0,1){4}}
\put(6.,0.){\line(0,1){4}}
\put(1.35,3.35){\framebox(0.4,0.4){}}
\put(3.35,3.35){\framebox(0.4,0.4){}}
\put(0.3,2.25){\rule{0.2cm}{0.2cm}}
\put(2.5,1.5){\circle{0.4}}
\put(4.5,1.5){\circle{0.4}}
\put(5.5,0.5){\circle*{0.3}}
\end{picture}

 & 
   \unitlength=4mm
\begin{picture}(6,3)
\put(0.,0.){\line(1,0){6}}
\put(0.,1.){\line(1,0){6}}
\put(0.,2.){\line(1,0){6}}
\put(0.,3.){\line(1,0){6}}
\put(0.,4.){\line(1,0){6}}
\put(0.,0.){\line(0,1){4}}
\put(1.,0.){\line(0,1){4}}
\put(2.,0.){\line(0,1){4}}
\put(3.,0.){\line(0,1){4}}
\put(4.,0.){\line(0,1){4}}
\put(5.,0.){\line(0,1){4}}
\put(6.,0.){\line(0,1){4}}
\put(0.35,3.35){\framebox(0.4,0.4){}}
\put(1.35,3.35){\framebox(0.4,0.4){}}
\put(3.35,3.35){\framebox(0.4,0.4){}}
\put(4.3,2.25){\rule{0.2cm}{0.2cm}}
\put(2.5,1.5){\circle{0.4}}
\put(5.5,0.5){\circle*{0.3}}
\end{picture}
 \\
   $u=332321$ & $v=342421$ & $w=442431$ \\
   \\
\unitlength=4mm
\begin{picture}(6,6)
\put(0.,0.){\line(1,0){6}}
\put(0.,1.){\line(1,0){6}}
\put(0.,2.){\line(1,0){6}}
\put(0.,3.){\line(1,0){6}}
\put(0.,4.){\line(1,0){6}}
\put(0.,5.){\line(1,0){6}}
\put(0.,6.){\line(1,0){6}}
\put(0.,0.){\line(0,1){6}}
\put(1.,0.){\line(0,1){6}}
\put(2.,0.){\line(0,1){6}}
\put(3.,0.){\line(0,1){6}}
\put(4.,0.){\line(0,1){6}}
\put(5.,0.){\line(0,1){6}}
\put(6.,0.){\line(0,1){6}}
\put(3.3,5.25){\rule{0.2cm}{0.2cm}}
\put(1.3,4.25){\rule{0.2cm}{0.2cm}}
\put(0.3,3.25){\rule{0.2cm}{0.2cm}}
\put(4.5,2.5){\circle{0.4}}
\put(2.5,1.5){\circle{0.4}}
\put(5.5,0.5){\circle*{0.3}} 
\end{picture}
& 
\unitlength=4mm
\begin{picture}(6,6)
\put(0.,0.){\line(1,0){6}}
\put(0.,1.){\line(1,0){6}}
\put(0.,2.){\line(1,0){6}}
\put(0.,3.){\line(1,0){6}}
\put(0.,4.){\line(1,0){6}}
\put(0.,5.){\line(1,0){6}}
\put(0.,6.){\line(1,0){6}}
\put(0.,0.){\line(0,1){6}}
\put(1.,0.){\line(0,1){6}}
\put(2.,0.){\line(0,1){6}}
\put(3.,0.){\line(0,1){6}}
\put(4.,0.){\line(0,1){6}}
\put(5.,0.){\line(0,1){6}}
\put(6.,0.){\line(0,1){6}}
\put(1.35,4.35){\framebox(0.4,0.4){}}
\put(3.35,5.35){\framebox(0.4,0.4){}}
\put(0.3,3.25){\rule{0.2cm}{0.2cm}}
\put(4.5,2.5){\circle{0.4}}
\put(2.5,1.5){\circle{0.4}}
\put(5.5,0.5){\circle*{0.3}} 
\end{picture}
& 
\unitlength=4mm
\begin{picture}(6,6)
\put(0.,0.){\line(1,0){6}}
\put(0.,1.){\line(1,0){6}}
\put(0.,2.){\line(1,0){6}}
\put(0.,3.){\line(1,0){6}}
\put(0.,4.){\line(1,0){6}}
\put(0.,5.){\line(1,0){6}}
\put(0.,6.){\line(1,0){6}}
\put(0.,0.){\line(0,1){6}}
\put(1.,0.){\line(0,1){6}}
\put(2.,0.){\line(0,1){6}}
\put(3.,0.){\line(0,1){6}}
\put(4.,0.){\line(0,1){6}}
\put(5.,0.){\line(0,1){6}}
\put(6.,0.){\line(0,1){6}}
\put(1.35,4.35){\framebox(0.4,0.4){}}
\put(3.35,5.35){\framebox(0.4,0.4){}}
\put(0.35,3.35){\framebox(0.4,0.4){}}
\put(4.3,2.25){\rule{0.2cm}{0.2cm}}
\put(2.5,1.5){\circle{0.4}}

\put(5.5,0.5){\circle*{0.3}} 
\end{picture}\\

 \multicolumn{3}{c}{$\pi=452631$}
\end{tabular}
\end{center}
\caption{\label{fig_Exp}
The permutation 
$\pi=452631\in\s_6$ in Figure \ref{pis}(a) is the expansion of each of $u$, $v$ and $w$.
We have that $\wpart(u)=\{\{1\},\{2,3\},\{4,5,6\}\}$, 
$\wpart(v)=\{\{1\},\{2,3\},\{4\},\{5,6\}\}$ and
$\wpart(w)=\{\{1\},\{2\},\{3\},\{4,5,6\}\}$; also, 
$\ppart(\pi)=\{\{1\},\{2,3\},\{4,5,6\}\}$.
}
\end{figure}

\section{Extension of $p$ to words}\label{ext-of-p-to-words}

In this section, we show that 
the statistic $\stat$ on words defined by us in Section~\ref{sec2} is equidistributed with the statistic 
$\maj$ on words, and thus our $\stat$ is Mahonian. In fact, we show a more general result on 
joint equidistribution of six statistics on words: see Theorem~\ref{thm1} for the case of dense 
words, and Theorem~\ref{thm2} for the case of arbitrary words. 

To this end, we extend the bijection $p$ from permutations to words, which is roughly 
done by the following three steps: For a word $v$, we apply the expansion 
operation ($\Exp$ defined above) in order to obtain a permutation $\pi$,
then apply the bijection $p$ on permutations to obtain $\sigma=p(\pi)$, and, finally, 
apply the inverse of the expansion operation to $\sigma$. 
The resulting word $w$ is the image of
$v$ by the extension of $p$ to words.
The main difficulty consists in the third step, since with no additional constraints, 
the expansion operation is not invertible. The main ingredient to overcome this,
is the bijection $\psi$ defined in Section \ref{def_psi}, which works due to
a consequence of Lemma \ref{le_th} expressed in 
Corollary \ref{cor}.

\subsection{Dense words}

Here we will extend the bijection $p:\s_n\to\s_n$ to length $n$ dense words over $[q]$.
For a dense word $v$ we construct a dense word  $w$, 
and show that the transformation $v\ {\mapsto}\ w$ is a bijection which 
preserving certain properties of $p$.

Let $v$ be a dense word in $[q]^n$, $R=\wpart(v)$, and $\pi$ and $\sigma$
be the permutations defined by:
\begin{itemize}
\item $\pi=\Exp(v)$ with $P=\ppart(\pi)$, and 
\item $\sigma=p(\pi)$ with $S=\ppart(\sigma)$.
\end{itemize}

Now let $w=\flat_T(\sigma)$ with $T=\psi_{P,S}(R)$. 

Clearly, when $v$ is a permutation, then $w=\sigma=p(\pi)$, and so
the restriction of the mapping $v\mapsto w$ to permutations is 
equal to $p$, and by a slight abuse of notation we denote
this mapping by $p$.

\begin{Exam} Let $v=342421$ as in Figure \ref{fig_Exp},
with $R=\wpart(v)=\{\{1\},\{2,3\},\{4\},\{5,6\}\}$. Then
\begin{itemize}
\item $\pi=\Exp(v)=452631$ (see also Figure \ref{pis}(a))
      and $P=\ppart(\pi)=\{\{1\},\{2,3\},\{4,5,6\}\}$;
\item $\sigma=p(\pi)=431526$ is the permutation in 
      Figure \ref{pis}(c), and 
      $S=\ppart(\sigma)=\{\{1,2\},\{3\},\{4,5,6\}\}$.
\end{itemize}
With the previous values, $T=\psi_{P,S}(R)=\{\{1,2\},\{3\},\{4\},\{5,6\}\}$
(see Example \ref{P_S}) and $p(v)=w=\flat_T(\sigma)=321414$.
It is routine to check that  $\stat\, v=\maj\, w=7$, and $\stat\, w=\maj\, v=11$. 
Notice that, $w$ is not a rearrangement of $v$.
\end{Exam}

The following theorem is a words counterpart of Theorem 2.1 in \cite{Bur} 
endowed with $\ides$ statistic. In that theorem, $\adj$ is extended to dense words as follows:
for a word $v$, $\adj\,v$ is the number of positions $i$, $1\leq i\leq n$,
in the word $v'=v0$ such that $v'_i=v'_{i+1}+1$, and  $i$ is 
the leftmost position where the letter $v'_i$ occurs in $v'$, while
$i+1$ is the rightmost position where the letter $v'_{i+1}$ occurs in $v'$.

\begin{The}\label{thm1}
The function $p$ is a bijection from length $n$ dense words over $[q]$ 
into itself, and the $6$-tuple $(\adj,\des,\ides,\Fi,\maj,\stat)\,v$ is equal to 
$(\adj,\des,\ides,\Fi,\stat,\maj)\,p(v)$, for any dense word $v\in [q]^n$. 
\label{the_eq_w_d}
\end{The}
\begin{proof}
First, since $p:\s_n\to \s_n$ is an involution, and the inverse of
$\psi_{P,S}:I_n|P\rightarrow I_n|S$ is $\psi_{S,P}:I_n|S\rightarrow I_n|P$,
it follows that $p(p(v))=v$ for any dense words in $[q]^n$, and so $p$
is an involution (and thus a bijection).

Let now $v$ be a dense word in $[q]^n$, $w=p(v)$,
$\pi=\Exp(v)$ and $\sigma=p(\pi)$, as in the definition of the transformation $p$ on
words. It follows that $\sigma=\Exp(w)$, and by Fact~\ref{fact_3},
that $\stat\,v=\stat\,\pi=\maj\,\sigma=\maj\,w$. 
Also, since $p$ is an involution, we have $\maj\,v=\stat\,w$.

In the word $v$, $v_i$ is a descent if and only if $\pi_i$ is a descent in 
$\pi$, and analogously for $w$
and $\sigma$.
Since $p$ preserves the number of descents on permutations, so it does on 
words.

Similarly, in the word $v$, $v_i$ is an occurrence of an $\ides$ if and only if 
$\pi_i$ is an occurrence of an $\ides$ in $\pi$, and analogously for $w$
and $\sigma$.
By Lemma \ref{le_th},
$p$ preserves the number of $\ides$'s on permutations, and thus so does on words.

The proof is similar for $\adj\,v=\adj\,p(v)$.
 
By the definition of $\ppart$ and the construction of 
$p$, $\ppart\, \pi$ and $\ppart\, \sigma$ are both
refinements of $\{\{1,2,\ldots,\pi_1-1\},\{\pi_1,\ldots,n\}\}$. In addition,
by Lemma \ref{le_by}, the number of `sub-parts' of 
$\{1,2,\ldots,\pi_1-1\}$ of $\ppart\, \pi$ and of $\ppart\, \sigma$
are the same.
Since $\wpart\, v$ is a refinement of $\ppart\, \pi$
having the same encoding as the refinement  $\wpart\, w$ of $\ppart\, \sigma$,
it follows that the number of `sub-parts' of 
$\{1,2,\ldots,\pi_1-1\}$ of $\wpart\, v$ and of $\wpart\, w$
are the same, and so $v_1=w_1$, that is $\Fi\, v=\Fi\, w$. \end{proof}

\subsection{General words}

For a word $v=v_1v_2\ldots v_n\in [q]^n$, we let $\red(v)$ denote the word obtained from 
$v$ in which the $i$th smallest letter in $v$ is substituted by $i$. 
For example, $\red(162414)=142313$.

Clearly, the function $\red$ produces a dense word and it is a bijection between the set of words over the alphabet from which $v$ is constructed and the set of dense words of the same length as that of $v$. Thus, to find the pre-image of  $\red(v)$, we need to know the alphabet from which $v$ is constructed. 

Since $\red$ and $p$ are bijections and $\red$ 
preserves the order on $[q]$, we have the following generalization 
of Theorem \ref{the_eq_w_d}, where by a slight abuse of notation,
we denote by $p$ the function $\red^{-1}\circ p\circ \red$, where $\red^{-1}$ uses the 
alphabet of the input word. Also, in the following theorem, we slightly abuse 
notation to denote by $\adj$ the composition $\adj\circ\red$. That is, to calculate the 
value of statistic $\adj$ on a given word $v\in  [q]^n$, one should first turn $v$ into 
the dense word $\red(v)$, and then calculate the value of $\adj$ using the definition state right before
Theorem~\ref{the_eq_w_d}.

\begin{The}\label{thm2}
The function $p$ is a bijection from $[q]^n$ into itself, and the $6$-tuple 
$$(\adj,\des,\ides,\Fi,\maj,\stat)\,v$$ is the same as  
$(\adj,\des,\ides,\Fi,\stat,\maj)\,p(v)$, for any word $v\in [q]^n$. 
\label{the_eq_w}
\end{The}

\section{Final remarks}

It is worth mentioning that our bijection $p$ does not preserve the number of occurrences 
of letters, while our computer experiments made us believe that 
such a bijection exists, and we invite the reader to find it. Also, it would be of interest to explore the property of being a Mahonian statistic on words for other Mahonian statistics on permutations defined in~\cite{BabSteim}.

Finally, a {\tt C} implementations of the bijection $p$ 
is on the web site of the second author \cite{Vaj_web}.

\section*{Acknowledgments}

The authors are grateful to the Edinburgh Mathematical Society for supporting the 
second author's visit of the University of Strathclyde, 
which helped this paper to appear.

\end{document}